\documentclass[a4paper,12pt]{amsart}
\usepackage{amssymb}
\usepackage{amsmath}
\usepackage{ascmac}
\usepackage{comment}
\usepackage{wrapfig}
\usepackage{txfonts}%肉厚な Times 系フォント
\newtheorem{thm}{Theorem}[section]
\newtheorem{prop}[thm]{Proposition}
\newtheorem{lem}[thm]{Lemma}

\newtheorem{rem}[thm]{Remark}

\def\o#1{\overline{#1}}
\def\u#1{\underline{#1}}

\def\dfrac#1#2{{\displaystyle\frac{#1}{#2}}}
\makeatletter
\@addtoreset{equation}{section}

\makeatother
\title{Lax pairs for additive difference Painlev\'e equations}
\author{Hidehito Nagao}
\address{Department of Arts and Science, National Institute of Technology, Akashi College, Hyogo 674-8501, Japan}
\email{nagao@akashi.ac.jp}

\begin{document}

\begin{abstract}
A Lax pair for the additive difference Painlev\'e equation of type $E_7^{(1)}$ is explicitly obtained as certain linear difference equations of scalar form. The compatibility of the Lax pair is proved by using certain characterization of the coefficients in the Lax equation. Some Lax pairs for types $E_6^{(1)}$, $D_4^{(1)}$ and $A_3^{(1)}$ are also given by the degeneration. 
\end{abstract}

\maketitle
%\tableofcontents
\renewcommand\baselinestretch{1.2}%行間調整

%%%%%%%%%%%%%%%%%%%%%%%%%%%
\section{Introduction}\label{sec:intro}
 
%\subsection{Sakai's classification of discrete Painlev\'e equations}　\\
The second order discrete Painlev\'e equations were classified in \cite{Sakai01}, based on rational surfaces connected to extended affine Weyl groups. There exist three types of discrete Painlev\'e equations in the classification: elliptic difference ($e$-), multiplicative difference ($q$-) and additive difference ($d$-) types. For each type of differences, the list of types of affine Weyl groups arising as the symmetry of B\"acklund transformations are given as the following diagram:

{\arraycolsep=0.1pt
%\[
\begin{equation}\label{eq:classification}\nonumber
\begin{array}{lcccccccccccccccccc}
{\rm elliptic} \hspace{1mm} \mbox{($e$-)}\quad&E^{(1)}_8\\[-2mm]
&&&&&&&&&&&&&&&{\mathbb Z}\\[-2mm]
&&&&&&&&&&&&&&\nearrow&&\searrow\\
{\rm multiplicative}\hspace{1mm} \mbox{($q$-)}\quad&E^{(1)}_8&\rightarrow&E^{(1)}_7&\rightarrow&E^{(1)}_6
&\rightarrow &D_5^{(1)}&\rightarrow&A_4^{(1)}
&\rightarrow&(A_2+A_1)^{(1)}&\rightarrow&(A_1+A_1^{\prime})^{(1)}
&\rightarrow&A_1^{(1)}&&{\mathcal D}_6\\[2mm]
\\
{\rm additive}\hspace{1mm} \mbox{($d$-)}\quad&E^{(1)}_8&\rightarrow&E^{(1)}_7&\rightarrow&E^{(1)}_6
&&\rightarrow&&D_4^{(1)}&\rightarrow&A_3^{(1)}&\rightarrow
&(A_1+A_1^{\prime})^{(1)}&\rightarrow&A^{(1)}_1&\rightarrow&{{\mathbb Z}_2}\\
&&&&&&&&&&&&\searrow&&\searrow&&&\downarrow\\
&&&&&&&&&&&&&A_2^{(1)}&\rightarrow&A_1^{(1)}\rightarrow&&1.\\
\end{array}
%\]
\end{equation}
}
Here $A \to B$ means that $B$ is obtained from $A$ by degeneration.

The main subject of this paper is the $d$-$E_7^{(1)}$ equation, namely the $d$-Painlev\'e equation of type $E_7^{(1)}$ \cite{GR99, KNY15, Sakai07}. For variables $f, g$ in $\mathbb{P}^{1}\times\mathbb{P}^{1}$ and parameters $\delta, a_i, b_i, (i=1,2,3,4) \in \mathbb{C}$ with a constraint
\begin{equation}\label{eq:cons}
\begin{array}l
a_1-a_2-a_3+a_4+b_1+b_2-b_3+b_4+\delta=0, 
\end{array}
\end{equation}
define a shift operator $T$ as  
\begin{equation}\label{eq:E7T}
T: \left(
\begin{array}{cccccc}
a_1,&a_2,&a_3,&a_4\\[3mm]
b_1,&b_2,&b_3,&b_4
\end{array},f,g
\right)\mapsto
\left(
\begin{array}{cccc}
a_1-\delta,&a_2-\delta,&a_3-\delta,&a_4-\delta\\[3mm]
b_1,&b_2,&b_3,&b_4
\end{array},\o{f},\o{g}
\right).
\end{equation}
Here for any object $X$ the corresponding shifts are denoted as $\o{X}:=T(X)$ and $\u{X}:=T^{-1}(X)$. Then the $d$-$E_7^{(1)}$ equation can be described by the birational transformation $T^{-1}(g)=\u{g}(f,g)$ and $T(f)=\o{f}(f,g)$ in $\mathbb{P}^1 \times\mathbb{P}^1$ as follows:

%We note that equations (\ref{eq:E7gd}) and (\ref{eq:E7fu}) were derived from certain Pad\'e problems (see Section \ref{sec:cr}). 
\begin{equation}\label{eq:E7gd}
\begin{array}l
\dfrac{\left(f-g-v\right)(f-\u{g}-v-\delta)}{(f-g)(f-\u{g})}=\dfrac{A(f)}{B(f)},
\end{array}
\end{equation}
\begin{equation}\label{eq:E7fu}
\begin{array}l
\dfrac{(f-g-v)(\o{f}-g-v+\delta)}{(f-g)(\o{f}-g)}=\dfrac{A(g+v)}{B(g)},
\end{array}
\end{equation}
where
\begin{equation}\label{eq:ABv}
A(x):=\prod_{i=1}^4(x-a_i), \quad B(x):=\prod_{i=1}^4(x-b_i),\quad v:=a_1+a_4-b_3.
\end{equation}

%\subsection{Background on Lax pairs for discrete Painlev\'e equations}　\\
We briefly recall a background on Lax pairs of discrete Painlev\'e equations. In \cite{JS96} the 2 $\times$ 2 matrix Lax pairs for type $q$-$D_5^{(1)}$ was derived making use of the connection preserving deformation of a 2 $\times$ 2 matrix system of $q$-difference equations. Some 2 $\times$ 2 matrix Lax pairs for types from $d$-$D_4^{(1)}$ to $d$-$A_1^{(1)}$ were derived in \cite{GORS98}, using a Schlesinger transformation of differential equations. In \cite{Sakai06} the 2 $\times$ 2 matrix Lax pairs for type $q$-$E_6^{(1)}$ was derived, making the use of the similar way as \cite{JS96}. Some 2 $\times$ 2 matrix Lax pairs were obtained utilizing moduli spaces of difference connection on $\mathbb{P}^1$ in \cite{AB06} for types $d$-$E_6^{(1)}$ and $d$-$D_4^{(1)}$. In \cite{Murata09} the 2 $\times$ 2 matrix Lax pairs for types from $q$-$A_4^{(1)}$ to $q$-$A_1^{(1)}$ were derived by utilizing the similar way as \cite{JS96, Sakai06}. Certain matrix Lax pairs were obtained as certain Fuchsian system of differential equations in \cite{Boalch09} for types $d$-$E_8^{(1)}$, $d$-$E_7^{(1)}$ and $d$-$E_6^{(1)}$. Some scalar Lax pairs of discrete Painlev\'e equations were given as linear difference equations, using a characterization in the coordinates $(f,g) \in \mathbb{P}^1 \times \mathbb{P}^1$ in \cite{Yamada09-2} for type $e$-$E_8^{(1)}$, \cite{Yamada11} for types $q$-$E_8^{(1)}$ and $q$-$E_7^{(1)}$, respectively. In \cite{KNY15} the scalar Lax pairs for all the discrete Painlev\'e equations have recently proposed by utilizing the characterization in the coordinates $(f,g)$. We remark that the results of this paper have been obtained independent of \cite{KNY15} and they were derived from certain Pad\'e problems (see Section \ref{sec:cr}). 

%We briefly recall a background on Lax pairs of discrete Painlev\'e equations in order of type. The scalar Lax pairs of discrete Painlev\'e equations were given as linear difference equations, using a characterization in the coordinates $(f,g) \in \mathbb{P}^1 \times \mathbb{P}^1$ in \cite{Yamada09-2} for type $e$-$E_8^{(1)}$, \cite{Yamada11} for types $q$-$E_8^{(1)}$ and $q$-$E_7^{(1)}$, respectively. The matrix Lax pairs were derived making use of the connection preserving deformation of a 2 $\times$ 2 matrix system of $q$-difference equations in \cite{Sakai06} for type $q$-$E_6^{(1)}$, \cite{JS96} for type $q$-$D_5^{(1)}$, and \cite{Murata09} for types from $q$-$A_4^{(1)}$ to $q$-$A_1^{(1)}$, respectively. The matrix Lax pairs were obtained as certain Fuchsian system of differential equations in \cite{Boalch09} for types $d$-$E_8^{(1)}$, $d$-$E_7^{(1)}$ and $d$-$E_6^{(1)}$. The matrix Lax pairs were obtained utilizing moduli spaces of difference connection in \cite{AB06} for types $d$-$E_6^{(1)}$ and $d$-$D_4^{(1)}$. The matrix Lax pairs were derived in \cite{GORS98} for types from $d$-$D_4^{(1)}$ to $d$-$A_1^{(1)}$, using a Schlesinger transformation of a 2 $\times$ 2 matrix system of differential equations. The scalar Lax pairs for all the discrete Painlev\'e equations have recently proposed by utilizing the characterization in the coordinates $(f,g)$ in \cite{KNY15}. We remark that the results of this paper have been obtained independent of \cite{KNY15} and the were derived from certain Pad\'e problems (see Section \ref{sec:cr}).
% 
This paper is organized as follows: In Section \ref{subsec:Lax} the Lax pair for type $d$-$E_7^{(1)}$ is explicitly obtained as certain linear difference equations of scalar form. In Section \ref{subsec:proof} we prove that the $d$-$E_7^{(1)}$ equation (\ref{eq:E7gd}) and (\ref{eq:E7fu}) is the sufficient condition for the compatibility of the Lax pair by using certain characterization in terms of $x$ (see Section \ref{subsec:proof}), which are related but different from the characterization in the coordinates $(f,g)$ (see Remark \ref{rem:fg}). In Section \ref{sec:Degenerations} the Lax pairs for types $d$-$E_6^{(1)}$, $d$-$D_4^{(1)}$ and $d$-$A_3^{(1)}$ are obtained by the degeneration. In Section \ref{sec:cr} the relation between Pad\'e interpolation problems and the results of this paper are discussed shortly as Concluding Remarks. 
%%%%%%%%%%%%%%%
\section{The Lax pair for $d$-Painlev\'e equation of type $E_7^{(1)}$ }\label{sec:E7}　
In this section, we will consider the Lax pair for type $d$-$E_7^{(1)}$. In Section \ref{subsec:Lax} we will give certain two linear difference equations as the scalar Lax pair for type $d$-$E_7^{(1)}$ and derive the $d$-$E_7^{(1)}$ equation. In Section \ref{subsec:proof} it will be proved that equations (\ref{eq:E7gd}) and (\ref{eq:E7fu}) are the sufficient condition for the compatibility of the linear equations. 

\subsection{Lax equations}\label{subsec:Lax}　\\
Let us consider two linear equations for an unknown function $y(x)$: $L_2(x)=0$ as the equation between $y(x)$, $y(x+\delta)$, $\o{y}(x)$ and $L_3(x)=0$ as the equation between $y(x)$, $\o{y}(x)$, $\o{y}(x-\delta)$, where
%Let us consider the following two linear equations for an unknown function $y(x)$. Take $L_2(x)$ as the equation between $y(x)$, $y(x+\delta)$, $\o{y}(x)$ and $L_3$ as the equation between $y(x)$, $\o{y}(x)$, $\o{y}(x-\delta)$.
$L_2(x)$ and $L_3(x)$ are given as linear three term expressions
\begin{equation}\label{eq:E7L2}
\begin{array}l   
L_2(x):=(x-f) \o{y}(x)-(x-g-v) y(x+\delta)+(x-g) y(x),
\end{array}
\end{equation}
\begin{equation}\label{eq:E7L3}
\begin{array}l
L_3(x):= w(x-\o{f}-\delta)+A(x)(x-g-\delta)\o{y}(x)-B(x-\delta)(x-g-v)\o{y}(x-\delta),
\end{array}
\end{equation}
and $f, g, \o{f}, w$ are variables 
%depending on parameters $a_i, b_i$ but 
independent of $x$. Then we have
\begin{prop}\label{prop:E7eq}
The compatibility of the linear equations $L_2=0$ (\ref{eq:E7L2}) and $L_3=0$ (\ref{eq:E7L3}) gives conditions (\ref{eq:E7gd}) and (\ref{eq:E7fu}). 
%These are equivalent to the $d$-Painlev\'e equation of type $E_7^{(1)}$ \cite{GR99, KNY15, Sakai07}.
\end{prop}

\begin{proof} 
Under the condition $x=f$, eliminating $y(x)$ and $y(x+\delta)$ from $L_2(x)=\u{L}_3(x+\delta)=0$, one obtains equation (\ref{eq:E7gd}). For $x=g$, eliminating $y(x+\delta)$ and $\o{y}(x)$ from $L_2(x)=L_3(x+\delta)=0$, we have the relation
\begin{equation}\label{eq:w1}
w=\dfrac{v(v-\delta)B(g)}{(f-g)(\o{f}-g)}.
\end{equation} 
In case of $x=g+v$, eliminating $y(x)$ and $\o{y}(x)$ from $L_2(x)=L_3(x)=0$, we have the relation
\begin{equation}\label{eq:w2}
w=\dfrac{v(v-\delta)A(g+v)}{(f-g-v)(\o{f}-g-v+\delta)}.
\end{equation} 
Hence, eliminating $w$ from relations (\ref{eq:w1}) and (\ref{eq:w2}), we obtain equation (\ref{eq:E7fu}).
\end{proof}

Let us consider the linear three term equation for the unknown function $y(x)$: $L_1(x)=0$ as the equation between $y(x+\delta), y(x), y(x-\delta)$.
%\begin{lem}\label{lem:E7L1}
Eliminating $\o{y}(x)$ and $\o{y}(x-\delta)$ from $L_2(x)=L_2(x-\delta)=L_3(x)=0$ (\ref{eq:E7L2}), (\ref{eq:E7L3}), 
one has the equation $L_1(x)=0$
where 
\begin{equation}\label{eq:E7L1}
\begin{array}l
L_1(x):= \dfrac{A(x)}{x-f}y(x+\delta)+\dfrac{B(x-\delta)}{x-f-\delta}y(x-\delta)\\[5mm]
\phantom{L_1(x): }-\dfrac{1}{x-g-v}\Big[\dfrac{A(x)(x-g)}{x-f}+\dfrac{V(x-\delta)}{(x-f-\delta)(x-g-\delta)}\Big]y(x),
\end{array}
\end{equation}
\begin{equation}\label{eq:V}
V(x):=B(x)(x-g-v)(x-g-v+\delta)-w(x-f)(x-\o{f}).
\end{equation}
%and $w$ is the same as in relation (\ref{eq:w1}). 
%\end{lem}

%\begin{coro} 
Eliminating $\o{f}$ and $w$ from the equation $L_1=0$ (\ref{eq:E7L1}) by using equation (\ref{eq:E7fu}) and relation (\ref{eq:w1}), the expression $L_1$ can be rewritten as
\begin{equation}\label{eq:E7LL1}
\begin{array}l   
L_1(x)=\displaystyle \frac{A(x)}{x-f}\left[y(x+\delta)-\frac{x-g}{x-g-v} y(x)\right]+\frac{B(x-\delta)}{x-f-\delta}
   \left[y(x-\delta)-\frac{x-g-v-\delta}{x-g-\delta} y(x)\right]\\[5mm]
\phantom{L_1(x): } \displaystyle +v \left[
   \frac{B(g)}{(f-g)
   (x-g-\delta)}-\frac{A(g+v)}{\left(f-g-v\right)
   \left(x-g-v\right)} \right]y(x).
\end{array}
\end{equation}
%\end{coro}
The linear difference equations $L_1=0$ (\ref{eq:E7LL1}) and $L_2=0$ (\ref{eq:E7L2}) can be regarded as the Lax pair for type $d$-$E_7^{(1)}$ and in Section \ref{subsec:proof} it will be proved that the $d$-$E_7^{(1)}$ equation (\ref{eq:E7gd}), (\ref{eq:E7fu}) is the sufficient condition for the compatibility of $L_1=0$ and $L_2=0$. The equation $L_1=0$ is equivalent to the scalar Lax equation in \cite{KNY15} by using a suitable gauge transformation of $y(x)$. On the other hand, a $4 \times 4$ matrix Lax pair was given as certain Fuchsian system of differential equations in \cite{Boalch09}.

\begin{rem}\label{rem:fg}
The expression $(f-g)(f-g-v)(x-f)(x-f-\delta)L_1(x)$ (\ref{eq:E7LL1}) has the characterization as a curve in $\mathbb{P}^1 \times \mathbb{P}^1$ with respect to the coordinates $(f, g) \in \mathbb{P}^1 \times \mathbb{P}^1$ as follows:\\
(i) The expression $(f-g)(f-g-v)(x-f)(x-f-\delta)L_1(x)$ is  a polynomial of bidegree $(3,2)$ in $(f,g)$.\\
(ii) As a polynomial, the expression $(f-g)(f-g-v)(x-f)(x-f-\delta)L_1$ vanishes at the the following $12$ points $(f_i, g_i) \in  \mathbb{P}^1 \times \mathbb{P}^1$ $(i=1,\ldots,12)$:
\begin{equation}
\begin{array}l
(b_i, b_i)_{i=1}^{4}, (a_i,a_i-v)_{i=1}^{4}, (x,x), (x-\delta,x-v-\delta), (x, \frac{(x-v)y(x+\delta)-xy(x)}{y(x+\delta)-y(x)}), (x-\delta, \frac{(x-v-\delta)y(x)-(x-\delta)y(x-\delta)}{y(x)-y(x-\delta)}).
\end{array}
\end{equation}
Moreover, the polynomial $(f-g)(f-g-v)(x-f)(x-f-\delta)L_1$ is characterized by these properties up to multiplicative constant.
\end{rem}

We note that similar characterization of the equation $L_1=0$ for the discrete Painlev\'e equations have been already given in \cite{KNY15, Yamada09-2, Yamada11}. In Section \ref{subsec:proof} we gives the different characterization of the coefficients of $y(x+\delta), y(x), y(x-\delta)$  in terms of $x$. 

\subsection{Proof of the compatibility condition of the Lax pair}\label{subsec:proof}　\\
%This subsection is based on the energetic discussion with Y.Yamada. 
In Section \ref{subsec:Lax} we derived the $d$-$E_7^{(1)}$ equation (\ref{eq:E7gd}) and (\ref{eq:E7fu}) as the necessary condition for the compatibility of the Lax pair (\ref{eq:E7LL1}) and (\ref{eq:E7L2}). In this subsection, we will prove that the $d$-$E_7^{(1)}$ equation is the sufficient condition for the compatibility of the Lax pair. 

%The equation $L_1$ (\ref{eq:E7L1}) has the following characterization.
\begin{lem}\label{lem:E7L1p}
The expression $(x-f)(x-f-\delta)L_1(x)$ (\ref{eq:E7L1}) (or (\ref{eq:E7LL1})) has the following characterization:\\
\hspace{5mm}(i) The coefficients of $y(x+\delta), y(x), y(x-\delta)$ are polynomials of degree $5$ in $x$.\\
\hspace{5mm}(ii) Under the condition 
\begin{equation}\label{eq:cond}
\begin{array}l
\displaystyle\frac{y(x+\delta)}{y(x)}=1+\frac{v}{x}+\dfrac{v(v-\delta)/2}{x^2}+O\Big(\frac{1}{x^3}\Big), \quad \frac{y(x-\delta)}{y(x)}=1-\frac{v}{x}+\frac{v(v-\delta)/2}{x^2}+O\Big(\frac{1}{x^3}\Big),
\end{array}
\end{equation}
the terms $x^5, \ldots, x$ in the expression $(x-f)(x-f-\delta)L_1(x)$ vanish, namely $(x-f)(x-f-\delta)L_1(x)=O(x^0)$ around $x=\infty$.\\
\hspace{5mm}(iii) The equation $(x-f)(x-f-\delta)L_1=0$ holds at the two points $x=f, f+\delta$ where
\begin{equation}\label{eq:y}
\dfrac{y(f+\delta)}{y(f)} =\dfrac{f-g}{f-g-v}.
% \quad \mbox{for}\quad x=f, f-\delta.
\end{equation}
Moreover, the coefficient of $y(x)$ in the polynomial $(x-f)(x-f-\delta)L_1(x)$ is uniquely characterized by these properties, provided that the coefficients of $y(x+\delta)$ and $y(x-\delta)$ are given as $(x-f-\delta)A(x)$ and $(x-f)B(x-\delta)$ respectively.
\end{lem}

We remark that two points $x=f, f+\delta$ are apparent singularities in the sense that at those two points the equation $(x-f)(x-f-\delta)L_1(x)=0$ is satisfied by the same condition (in this case (\ref{eq:y})).

\begin{proof}
The property (i) is given by computation using equations (\ref{eq:w1}) and (\ref{eq:w2}). Namely, the expression $\frac{V(x-\delta)}{x-g-\delta}$ reduces to a polynomial of degree 5 in $x$ under condition (\ref{eq:w1}). Furthermore the coefficient of the term $y(x)$ is given as a polynomial of degree 5 in $x$ by using condition (\ref{eq:w2}). The property (ii) can easily be confirmed by condition (\ref{eq:cond}). The property (iii) follows by substituting $x=f, f-\delta$ into the equation $L_1(x)=0$. 
\end{proof}

Eliminating ${y}(x)$ and ${y}(x+\delta)$ from $L_2(x)= L_3(x)=L_3(x+\delta)=0$ (\ref{eq:E7L2}), (\ref{eq:E7L3}), one has the linear three term equation $L_1^*(x)=0$ for 
$\o{y}(x+\delta), \o{y}(x), \o{y}(x-\delta)$ where
\begin{equation}\label{eq:E7L1*}
\begin{array}l
L_1^*(x):= \dfrac{\o{A}(x)}{x-\o{f}}\o{y}(x+\delta)+\dfrac{\o{B}(x-\delta)}{x-\o{f}-\delta}\o{y}(x-\delta)\\[5mm]
\phantom{L_1(x): }-\dfrac{1}{x-g-v}\Big[\dfrac{A(x)(x-g-\delta)}{x-\o{f}-\delta}+\dfrac{V(x)}{(x-\o{f})(x-g)}\Big]\o{y}(x).
\end{array}
\end{equation}
%and $V(x)$ is defined by (\ref{eq:V}).

The following Lemma (and its proof) is similar to Lemma \ref{lem:E7L1p}.
%The linear equation (\ref{eq:E7L1*}) has the following characterization.
\begin{lem}\label{lem:E7L1*p}
The expression $(x-\o{f})(x-\o{f}-\delta)L_1^*(x)$ (\ref{eq:E7L1*}) has the following characterization:\\
\hspace{5mm}(i) The coefficients of $\o{y}(x+\delta), \o{y}(x), \o{y}(x-\delta)$ are polynomials of degree $5$ in $x$.\\
\hspace{5mm}(ii) Under the condition 
\begin{equation}\label{eq:cond*}
\begin{array}l
\displaystyle\frac{\o{y}(x+\delta)}{\o{y}(x)}=1+\frac{\o{v}}{x}+\dfrac{\o{v}(\o{v}-\delta)/2}{x^2}+O\Big(\frac{1}{x^3}\Big), \quad \frac{\o{y}(x-\delta)}{\o{y}(x)}=1-\frac{\o{v}}{x}+\frac{\o{v}(\o{v}-\delta)/2}{x^2}+O\Big(\frac{1}{x^3}\Big),
\end{array}
\end{equation}
the terms $x^5, \ldots, x$ in the polynomial $(x-\o{f})(x-\o{f}-\delta)L_1(x)$ vanish, namely $(x-\o{f})(x-\o{f}-\delta)L_1^*(x)=O(x^0)$ around $x=\infty$.\\
\hspace{5mm}(iii) The equation $(x-\o{f})(x-\o{f}-\delta)L_1^*=0$ holds at the two points $x=\o{f}, \o{f}+\delta$ where
\begin{equation}\label{eq:yu}
\dfrac{\o{y}(\o{f}+\delta)}{\o{y}(\o{f})} =\dfrac{\o{B}(\o{f})(\o{f}-g-v+\delta)}{\o{A}(\o{f})(\o{f}-g)}.
% \quad \mbox{for}\quad x=\o{f}, \o{f}-\delta.
\end{equation}
Moreover, the coefficient of $\o{y}(x)$ in the polynomial $(x-\o{f})(x-\o{f}-\delta)L_1^*(x)$ is uniquely characterized by these properties, provided that the coefficients of $\o{y}(x+\delta)$ and $\o{y}(x-\delta)$ are given as $(x-\o{f}-\delta)\o{A}(x)$ and $(x-\o{f})\o{B}(x-\delta)$ respectively.
\end{lem}
%\begin{proof}
%This proof is similar to that of Lemma \ref{lem:E7L1p}.
%\end{proof}

The following is the main result of this paper.
\begin{thm}\label{thm:Lax}
The linear equations $L_1=0$ (\ref{eq:E7LL1}) and $L_2=0$ (\ref{eq:E7L2}) for the  unknown function $y(x)$ are compatible
if and only if the $d$-$E_7^{(1)}$ equation (\ref{eq:E7gd}) and (\ref{eq:E7fu}) are satisfied. 
\end{thm}

\begin{proof}
The compatibility means that the shift operator $T$ changes the equation $L_1=0$ into the equation $L_1^{\ast}=0$, i.e. the commutativity of  the following: 
\begin{equation}\label{eq:classification}\nonumber
\begin{array}{cccc}
L_1^*=0 \ (\mbox{Lemma} \  \ref{lem:E7L1*p})\quad &\Leftrightarrow&L_1^*=0\  (\ref{eq:E7L1*})\\
\uparrow&&\uparrow\\
\mbox{$T$-shift}\hspace{2mm}(\ref{eq:E7T})&&L_2= L_3=0 \hspace{2mm}(\ref{eq:E7L2}), (\ref{eq:E7L3})\\
\uparrow&&\downarrow\\
L_1=0 \ (\mbox{Lemma} \  \ref{lem:E7L1p})\ &\Leftrightarrow& L_1=0 \ (\ref{eq:E7L1}) &\Leftrightarrow L_1=0 \ (\ref{eq:E7LL1}).
\end{array}
\end{equation}
This commutativity is almost clear from the characterizations (i),(ii) of the equation $L_1=0$ (respectively $L_1^*=0$) in Lemma \ref{lem:E7L1p} (respectively Lemma \ref{lem:E7L1*p}).
The remaining task is to check that the operator $T$ changes expression (\ref{eq:y}) into expression (\ref{eq:yu}), utilizing the characterization (iii) of the equation $L_1=0$ (respectively $L_1^*=0$) and equation (\ref{eq:E7gd}).
\end{proof}

As the point of the proof, the following two are applied to type $d$-$E_7^{(1)}$ together: The first is that the equation $L_1(f,\o{f}, g)=0$ in terms of $f, \o{f}, g$ is derived from the equations $L_2(f,g)=0$ and $L_3(\o{f},g)=0$ (see \cite{NTY13}). The second is that the equation $L_1(f,\o{f}, g)=0$ is characterized as a polynomial in terms of $x$ (see \cite{NY16}).

%%%%%%%%%%%%%%%%%%%%%%%%
\section{Degenerations}\label{sec:Degenerations}　
In this section, we will consider degeneration limits from type $d$-$E_7^{(1)}$ to types $d$-$E_6^{(1)}$, $d$-$D_4^{(1)}$ and $d$-$A_3^{(1)}$.
%%%%%%%%%%%%%%%%%%%%%%%%
\subsection{Degeneration from type $d$-$E_7^{(1)}$ to type $d$-$E_6^{(1)}$}\label{subsubsec:E6}　\\
Degeneration from type $d$-$E_7^{(1)}$ to type $d$-$E_6^{(1)}$ is obtained by setting a transformation
\begin{equation}
\begin{array}l
a_4 \to -\dfrac{1}{\varepsilon},\quad a_3 \to -\dfrac{1}{\varepsilon}+a_1-a_2+b_1+b_2-b_3+b_4+\delta,\quad 
\dfrac{y(x+\delta)}{y(x)} \to \dfrac{\varepsilon y(x+\delta)}{y(x)},
%,\quad \dfrac{y(x-\delta)}{y(x)} \to \dfrac{y(x-\delta)}{y(x)\varepsilon},
\end{array}
\end{equation}
and taking the limit $\varepsilon \to 0$.\\
The time evolution is given by a shift operator
\begin{equation}\label{eq:E6T}
T: (a_1,a_2,b_1,b_2,b_3,b_4, f, g) \mapsto (a_1-\delta,a_2-\delta,b_1,b_2,b_3,b_4, \o{f}, \o{g}).
\end{equation}
The $d$-Painlev\'e $E_6^{(1)}$ equation \cite{KNY15, RGTT01, Sakai07} is equivalent to a birational transformation
\begin{equation}\label{eq:E6eq}
\begin{array}{l}
(f-g)(f-\u{g})=\dfrac{B(f)}{(f-a_1)(f-a_2)},\quad (f-g)(\o{f}-g)=\dfrac{B(g)}{(g+u_1)(g+u_2)}.
\end{array}
\end{equation}
where $B(x):=\prod_{i=1}^4(x-b_i)$, $u_1:=a_2-b_1-b_2-b_4-\delta$, and $u_2:=a_1-b_3$.\\
The Lax pair is linear difference equations
\begin{equation}\label{eq:E6Lax}
\begin{array}{l}
L_1(x):= \displaystyle \frac{(x-a_1)(x-a_2)}{x-f}\left[y(x+\delta)-(x-g)y(x)\right]+\frac{B(x)}{x-f-\delta}
   \left[y(x-\delta)-\frac{y(x)}{x-g-\delta}\right]\\[5mm]
\phantom{L_1(x):= } \displaystyle +\left[(g+u_1)(g+u_2)-\dfrac{B(x)}{(f-g)(x-g-\delta)} \right]y(x)=0,
\\[5mm]  
L_2(x):= (x-f) \o{y}(x)- y(x+\delta)+(x-g) y(x)=0.
\end{array}
\end{equation}The equation $L_1=0$ (\ref{eq:E6Lax}) is equivalent to the scalar Lax equation in \cite{KNY15} by using suitable gauge transformation of $y(x)$. We remark that a 2 $\times$ 2 matrix Lax pair was obtained utilizing moduli space of difference connection in \cite{AB06} and a $3 \times 3$ matrix Lax pair was given as the Fuchsian system of differential equations in \cite{Boalch09}.
%%%%%%%%%%%%%%%%%%%%%%%%
\subsection{Degeneration from type $d$-$E_6^{(1)}$ to type $d$-$D_4^{(1)}$}\label{subsec:D4}　\\
Degeneration from type $d$-$E_6^{(1)}$ to type $d$-$D_4^{(1)}$ is obtained by setting a transformation
\begin{equation}
\begin{array}l
b_3 \to -\dfrac{1}{\varepsilon},\quad b_4 \to -\dfrac{1}{\varepsilon t},\quad g \to -\dfrac{1}{g\varepsilon},\quad \dfrac{y(x+\delta)}{y(x)} \to \dfrac{y(x+\delta)}{\varepsilon y(x)},
%\quad \dfrac{y(x-\delta)}{y(x)} \to \dfrac{\varepsilon y(x-\delta)}{y(x)},
\quad \dfrac{\o{y}(x)}{y(x)} \to \dfrac{\o{y}(x)}{\varepsilon y(x)},
\end{array}
\end{equation}
and taking the limit $\varepsilon \to 0$.\\
The time evolution is given by a shift operator
\begin{equation}\label{eq:D4T}
T: (a_1,a_2,b_1,b_2,t,f,g) \mapsto (a_1-\delta,a_2-\delta,b_1,b_2,t,\o{f},\o{g}).
\end{equation}
The $d$-$D_4^{(1)}$ equation \cite{GORS98, KNY15, RGTT01, Sakai07} is equivalent to a birational transformation
\begin{equation}\label{eq:D4eq}
\begin{array}{l}
g\u{g}=\dfrac{t(f-a_1)(f-a_2)}{(f-b_1)(f-b_2)},\quad f+\o{f}-b_1-b_2+\dfrac{a_1}{g-1}+\dfrac{tu}{g-t}=0.
\end{array}
\end{equation}
where $u:=a_2-b_1-b_2-\delta$.\\
The Lax pair is linear difference equations
\begin{equation}\label{eq:D4Lax}
\begin{array}{l}
L_1(x):=\displaystyle \frac{t(x-a_1)(x-a_2)}{x-f}\left[y(x+\delta)-\frac{y(x)}{g}\right]+\frac{(x-b_1-\delta)(x-b_2-\delta)}{x-f-\delta}
   \left[y(x-\delta)-gy(x)\right]\\[5mm]
\phantom{L_1(x):=} \displaystyle +(g-1)(g-t)\left[\dfrac{a_1}{g-1}+\dfrac{x+f-a_1-a_2}{g}+\dfrac{u}{g-t} \right]y(x)=0,
\\[5mm]  
L_2(x):=(x-f) \o{y}(x)- y(x+\delta)+ \dfrac{y(x)}{g}=0.
\end{array}
\end{equation}
The equation $L_1=0$ (\ref{eq:D4Lax}) is equivalent to the scalar Lax equation in \cite{KNY15} by using suitable gauge transformation of $y(x)$. We remark that a 2 $\times$ 2 matrix Lax pair for type $d$-${\rm P_V}$ was obtained by utilizing a Schlesinger transformation of a differential equation in \cite{GORS98} and a 2 $\times$ 2 matrix Lax pair was obtained utilizing moduli space of difference connection in \cite{AB06}.

%%%%%%%%%%%%%%%%%%%%%%%%
\subsection{Degeneration from type $d$-$D_4^{(1)}$ to type $d$-$A_3^{(1)}$}
\label{subsec:A3}　\\
Degeneration from type $d$-$D_4^{(1)}$ to type $d$-$A_3^{(1)}$ is obtained by setting a transformation
\begin{equation}
\begin{array}l
a_2 \to -\dfrac{1}{\varepsilon},\quad t \to \varepsilon t,
\end{array}
\end{equation}
and taking the limit $\varepsilon \to 0$.\\
The time evolution is given by a shift operator
\begin{equation}\label{eq:A3T}
T: (a_1,b_1,b_2,t,f,g) \mapsto (a_1-\delta,b_1,b_2,t,\o{f},\o{g}).
\end{equation}
The $d$-$A_3^{(1)}$ equation \cite{GORS98,KNY15,RGTT01,Sakai07} is equivalent to a birational transformation
\begin{equation}\label{eq:A3eq}
\begin{array}{l}
g\u{g}=\dfrac{t(f-a_1)}{(f-b_1)(f-b_2)},\quad f+\o{f}-b_1-b_2-\dfrac{t}{g}+\dfrac{a_1}{g-1}=0.
\end{array}
\end{equation}
The Lax pair is linear difference equations
\begin{equation}\label{eq:A3Lax}
\begin{array}{l}
L_1(x):=\displaystyle \frac{t(x-a_1)}{x-f}\left[y(x+\delta)-\frac{y(x)}{g}\right]+\frac{(x-b_1-\delta)(x-b_2-\delta)}{x-f-\delta}
   \left[y(x-\delta)-gy(x)\right]\\[5mm]
\phantom{L_1(x):=} \displaystyle +(g-1)\left[x+f-b_1-b_2-\delta-\dfrac{t}{g}+\dfrac{a_1}{g-1} \right]y(x)=0,
\\[5mm]  
L_2(x):=(x-f) \o{y}(x)- y(x+\delta)+ \dfrac{y(x)}{g}=0.
\end{array}
\end{equation}
The equation $L_1=0$ (\ref{eq:A3Lax}) is equivalent to the scalar Lax equation in \cite{KNY15} by using suitable gauge transformation of $y(x)$. We remark that a 2 $\times$ 2 matrix Lax pair for type $d$-${\rm P_{IV}}$ was obtained by using a Schlesinger transformation of a differential equation in \cite{GORS98}.
%%%%%%%%%%%%%%%%%%%%%%%%
\section{Concluding Remarks}\label{sec:cr}　
As mentioned in Section \ref{sec:intro}, the results of this paper were obtained making use of Pad\'e interpolation problems. We will discuss the relation with Pad\'e problem and the results shortly.  

%First we give an explanation of "Pad\'e method".
There exists a simple method to study the Painlev\'e/Garnier equations by using Pad\'e approximation \cite{Yamada09-1}. In this method, one can obtain the evolution equation, the Lax pair and some special solutions simultaneously, 
starting from a suitable Pad\'e approximation (or interpolation) problem. Concretely, for a suitable given function $Y(x)$, we look for polynomials $P_m(x)$ and $Q_n(x)$ of degree $m$ and $n \in \mathbb{Z}_{\ge 0}$, satisfying the approximation condition:
\begin{equation}\label{eq:pade1}
Y(x)\equiv\dfrac{P_m(x)}{Q_n(x)} \quad (\mbox{mod}\hspace{1mm} x^{m+n+1}),
\end{equation}
or the interpolation condition:
\begin{equation}\label{eq:pade2}
Y(x_s)=\dfrac{P_m(x_s)}{Q_n(x_s)} \quad (s=0,\ldots, m+n).
\end{equation}
Define certain shift operator $T$ (e.g.(\ref{eq:E7T})) and consider two linear three term relations satisfied by the function $y(x)=P_m(x), Y(x)Q_n(x)$, such as the equations $L_2(x)=L_3(x)=0$ (e.g.(\ref{eq:E7L2}), (\ref{eq:E7L3})). The compatibility of the linear relations gives Painlev\'e equation (e.g.(\ref{eq:E7gd}), (\ref{eq:E7fu})). 
%The Lax equation $L_1(x)=0$ (e.g.(\ref{eq:E7L1})) can be derived from the contiguity relations and the Lax pair $L_1(x)=L_2(x)=0$ (e.g.(\ref{eq:E7L1}), (\ref{eq:E7L2})) is obtained. 
Moreover special solutions of the Painlev\'e equation can be obtained from the polynomials $P_m(x)$ and $Q_n(x)$ of the Pad\'e condition (\ref{eq:pade1}) (or (\ref{eq:pade2})).
%In other words the solutions $P_m(x)$ and $Q_n(x)$ of Pad\'e problem gives special solutions of the Lax pair and Painlev\'e equation. 
%The solutions $P_m(x)$ and $Q_n(x)$ are given as follows:  
%
%For any given function $Y(x)=\sum_{k=0}^{\infty}p_k x^k$, ($p_i=0, i<0$), the polynomials $P(x)$ and $Q(x)$ satisfying the approximation condition (\ref{eq:pade1}) are given by the expression
%\begin{equation}\label{eq:Schur}
%P(x)=\displaystyle\sum_{i=0}^{m}s_{(m^n, i)}x^i,\quad Q(x)=\displaystyle\sum_{i=0}^{n}s_{((m+1)^i, m^{n-i})}(-x)^i,
%\end{equation}
%where  $s_{\lambda}$ is the Schur function defined by the Jacobi Trudi formula $s_{(\lambda_1, \ldots, \lambda_l)}:=\det(p_{\lambda_i-i+j})_{i, j=1}^l$ and $m^n:=(\underbrace{m,m,\ldots,m}_n)$. The formulae (\ref{eq:Schur}) are explained in section 2 of \cite{Yamada09-1}. 
%
%The polinomials $P_m(x)$ and $Q_n(x)$ satisfying the interpolation condition (\ref{eq:pade2})
%are given by the determinant expressions
%\begin{equation}\label{eq:Jacobi}
%P_m(x)=F(x)\det\Big[\sum^{m+n} _{s=0}u_s \dfrac{x_s^{i+j}}{x-x_s }\Big]^n _{i,j=0}, \quad Q_n(x)=\det\Big[\sum^{m+n} _{s=0}u_s x_s^{i+j}(x-x_s )\Big]^{n-1} _{i,j=0},
%\end{equation}
%where $u_s :=\sfrac{Y(x_s)}{F^{\prime}(x_s)}$ and $F(x):=\prod_{i=0}^{m+n}(x-x_i )$. The formulae (\ref{eq:Jacobi}) are introduced in \cite{Jacobi}.
We call the method mentioned above "Pad\'e method". The Pad\'e method based on (\ref{eq:pade1}) and (\ref{eq:pade2}) have been applied in \cite{Ikawa13, Nagao15-2, NY16, Yamada09-1} and \cite{Ikawa13, Nagao15, NY16, NTY13, Yamada14} respectively.
%to various types of discrete Painlev\'e equations \cite{KNY15, RGTT01, Sakai07} and Garnier system \cite{Sakai05} before. 

The results of this paper can be obtained by the Pad\'e interpolation problems (\ref{eq:pade2}) on the additive ($\delta$-) grid $x_s=s\delta$. The interpolated sequences $Y_s:=Y(x_s)$ can be chosen as follows:
\begin{equation}\label{eq:Ylist}
\begin{tabular}{|c|c|c|c|c|}
\hline
&$d$-$E_7^{(1)}$&$d$-$E_6^{(1)}$&$d$-$D_4^{(1)}$&$d$-$A_3^{(1)}$\\
\hline
$\begin{array}{c}\\Y_s\\ \\\end{array}$&$\displaystyle\prod_{i=1}^3\dfrac{(\frac{b_i}{\delta})_s}{(\frac{a_i}{\delta})_s}$&
$\displaystyle\prod_{i=1}^2\dfrac{(\frac{b_i}{\delta})_s}{(\frac{a_i}{\delta})_s}$&
$c^s\displaystyle\dfrac{(\frac{b_1}{\delta})_s}{(\frac{a_1}{\delta})_s}$&
$d^s(\frac{b_1}{\delta})_s$\\
\hline
\end{tabular}
\end{equation}
Here parameters $a_i, b_i, c, d \in \mathbb{C}^{\times}$ and $a_1+a_2+a_3-b_1-b_2-b_3-\delta (m-n)=0$ is a constraint for the parameters in type $d$-$E_7^{(1)}$, and Pochhammer's symbol is defined by
\begin{equation}\label{eq:Poch}
\begin{array}{l}
(a_1, a_2, \cdots, a_i)_j:=\prod_{k=0}^{j-1}(a_1+k)(a_2+k)\cdots(a_i+k).
\end{array}
\end{equation}
The details will be presented in a forthcoming paper.
\section*{Acknowledgment}　
The author would like to express his gratitude to Professor Yasuhiko Yamada for fruitful discussion on this research.
%, and for his encouragement. 
He is also grateful to the organizer Saburou Kakei of the workshop "Theory of Integrable Systems and its Applications in Various Fields" (RIMS, August 2015), where the results of this note were presented.  
The author thanks to Professors Kenji Kajiwara, Tetsu Masuda, and Hidetaka Sakai for stimulating comments and for kindhearted support. The author would like to thank the referee for his valuable comments and suggestions.
%%%%%%%%%%%%%%%%%%%%%%

\end{document}